\documentclass[a4paper,12pt]{article}

\makeatletter
 
  \@addtoreset{equation}{section}
 \makeatother

\usepackage{amsmath,amssymb}
\usepackage{amscd}
\usepackage{mathrsfs}
\usepackage{amsthm}
\usepackage{color}

\usepackage{comment}

\theoremstyle{plain}

\newtheorem{theorem}{\bf Theorem}[section]
\newtheorem{lemma}[theorem]{\bf Lemma}
\newtheorem{proposition}[theorem]{\bf Proposition}
\newtheorem{corollary}[theorem]{\bf Corollary}

\theoremstyle{definition}
\newtheorem{definition}[theorem]{\bf Definition}
\newtheorem{example}[theorem]{\bf Example}
\newtheorem{remark}[theorem]{\bf Remark}

  \makeatletter
  \newcommand{\subsubsubsection}{\@startsection{paragraph}{4}{\z@}%
    {1.0\Cvs \@plus.5\Cdp \@minus.2\Cdp}%
    {.1\Cvs \@plus.3\Cdp}%
    {\reset@font\sffamily\normalsize}
  }
  \makeatother
\title{Time operators for continuous-time and discrete-time quantum walks
}

\author{Daiju Funakawa
       \thanks{Department of Electronics and Information Engineering, Hokkai-Gakuen University,
 		Sapporo 062-8605, Japan,
 	    e-mail: funakawa@hgu.jp},
 	Yasumichi Matsuzawa
 	\thanks{Department of Mathematics, Faculty of Education, Shinshu University,
        6-Ro, Nishi-nagano, Nagano 380-8544, Japan,
        e-mail: myasu@shinshu-u.ac.jp},
 	Akito Suzuki
 	\thanks{Division of Mathematics and Physics, 
 		Faculty of Engineering, Shinshu University, Wakasato, Nagano 380-8553, Japan, 
 		e-mail: akito@shinshu-u.ac.jp},\\
 	Itaru Sasaki
 	\thanks{Department of Mathematics, Shinshu University, Matsumoto 390-8621, Japan,
        e-mail: isasaki@shinshu-u.ac.jp},
 	Noriaki Teranishi
 	\thanks{Department of Mathematics, Hokkaido University, Sapporo 060-0810, Japan,
        e-mail: teranishi@math.sci.hokudai.ac.jp}
}
\date{}

\begin{document}
\maketitle

\begin{abstract}
We construct concrete examples of time operators
for both continuous and discrete-time homogeneous quantum walks, and we determine their deficiency indices
and spectra. For a discrete-time quantum walk, the time operator can be self-adjoint if the time evolution operator has a non-zero winding number. 
In this case, its spectrum becomes a discrete set of real numbers.
\end{abstract}

\section{Introduction}
Quantum walks are universal computational primitives
\cite{Childs1, Childs2, Lovett} and efficient tools
for building quantum algorithms \cite{Amb}. 
They can also simulate quantum systems such as 
Dirac particles \cite{Arnault, MS, Molf}. 
It is hence important to know the dynamics of the quantum walk. 
In particular, many authors have been studied the long-time behavior 
in various ways \cite{Fillman, HKSS, Konno1, Konno2, MS4, SS, TateSunada, Su16}. 
Here we employ time operators to obtain time decay estimates of transition
probabilities between states, which will play important roles in implementing
quantum walk based algorithms and simulations. 

A time operator of a Hamiltonian $H$ is formally defined as a Hermitian operator $T$ which satisfies the canonical commutation relation with $H$, i.e., $TH-HT=i$.
It was widely believed for a long time that in quantum theory there exists no time operator, as Pauli pointed out in his famous textbook \cite[p. 63, footnote 2]{Pa80}.
On the other hand, Aharonov and Bohm constructed a concrete time operator of the one-dimensional free Hamiltonian \cite{AB61}.
This apparently contradicts to Pauli's claim.
The key to solve the contradiction is that Pauli's argument can not apply to Araronov and Bohm's time operator because it is quite formal.
This suggests that we must study time operators in a mathematically rigorous way.
It is important to pay enough attention to domains of time operators.
Such a study was initiated by Miyamoto \cite{Mi01}.
He expand his theory in a functional analytic context.
One of the most important result therein is that a time operator is related to ``time'' in the sense that 
the decay of the transition probability is estimated by the amount which is determined only from the time operator \cite[Theorem 4.1]{Mi01}.

In this paper, we construct time operators for concrete examples of quantum walks as a first step.
There are two kinds of quantum walks, one of which is called a {\it continuous-time quantum walk}, 
and the other is called a {\it discrete-time quantum walk}.

In the case of the continuous-time quantum walk,
the time evolution is described by the Hamiltonian $H$ of the system. 
Hence we can use Miyamoto's theory of time operators of self-adjoint operators. 
More precisely, we construct a strong time operator $T$ of $H$,
which is a symmetric operator satisfying the following operator equality 
\[ e^{itH} T e^{-itH} = T + t, \quad t \in \mathbb{R}. \]
We calculate the deficiency indices and the spectrum of $T$. 
As a result, the spectrum of the strong time operator $T$ becomes the set of complex numbers,
and $T$ does not have any self-adjoint extension that
is a strong time operator. 

On the contrary to the continuous case, 
the time evolution of a discrete-time quantum walk is described by a unitary operator.  
Hence, we need a notion of time operators of unitary operators.
Sambou and Tiedra de Aldecoa \cite{ST} recently proposed such a notion, and they called a Hermitian operator $T$ with the commutation relation $TU-UT=U$ on some domain (possibly not the operator equality) a time operator of a unitary operator $U$. 
In this paper, we say that a symmetric operator $T$ is a strong time operator of $U$
if $T$ satisfies the following operator equality 
\[ U^*TU = T + 1. \] 
We construct strong time operators for one-dimensional discrete-time quantum walks, and calculate their deficiency indices and spectra. 
We then observe that the strong time operator can be self-adjoint if the time evolution operator has a non-zero winding number. 
In this case, its spectrum becomes a discrete set of real numbers.

The paper is organized as follows.
In Section 2, we first recall time operators of self-adjoint operators.
Next we introduce definitions of time operators of unitary operators.
Their basic properties are investigated here.
In Section 3, examples of strong time operators for continuous-time quantum walks are given.
In Section 4, we construct examples of strong time operators for discrete-time quantum walks. 
We also determine their deficiency indices and spectra by passing to the energy representations.
It turns out that they have self-adjoint extensions, and hence such systems have self-adjoint time operators. 
\\

\noindent
\underline{\bf Definitions and notations}\\
Let $\mathcal{H}$ be a complex Hilbert space. 
We denote the inner product and the norm of $\mathcal{H}$ by $\langle\cdot, \cdot\rangle$ (complex linear in the second variable) and $\| \cdot \|,$ respectively.

For an operator $A$, we denote by $D(A)$ the domain of $A$, and by $\sigma(A)$ the spectrum of $A$. 
An operator $A$ on $\mathcal{H}$ is said to be {\it densely defined} if $D(A)$ is dense in $\mathcal{H}$.

Let $A, B$ be two operators.
We write $A=B$ if $D(A)=D(B)$ and $A\psi=B\psi$ for any $\psi\in D(A)$.
We say that $B$ is an {\it extension} of $A$ and write $A\subset B$ if $D(A)\subset D(B)$ and $A\psi=B\psi$ for any $\psi\in D(A)$.
Thus $A=B$ if and only if $A\subset B$ and $B\subset A$.

A densely defined operator $A$ is said to be {\it symmetric} if $A\subset A^*$.
In addition, if $A=A^*$ holds, then $A$ is said to be {\it self-adjoint}.
We denote by $\bar{A}$ the closure of a closable operator $A$.
Any symmetric operators are closable.
A symmetric operator $A$ is called {\it essentially self-adjoint} if $\bar{A}$ is self-adjoint.
      
For two operators $A, B$, we denote by $[A,B]$ the commutator of $A$ and $B$, 
i.e., $D([A,B]):=D(AB)\cap D(BA)$ and $[A,B]\psi:=AB\psi-BA\psi$ for $\psi\in D([A,B])$.

For the reader who is not familiar with functional analysis, we refer to the excellent textbook \cite{Sc12}.

Finally, we denote by $C_0^{\infty}(\Omega)$ the set of infinitely many differentiable functions
with compact support contained in $\Omega\subset\mathbb{R}$.

\section{Definitions of time operators}
In this section, we first recall the definitions of two classes of time operators of a self-adjoint operator.
Next, we introduce definitions of time operators of a unitary operator, as unitary analogs of time operators of a self-adjoint operator.
We also discuss some basic properties of them. 
 
 	\subsection{Time operators of a self-adjoint operator}
 	Let $H$ be a self-adjoint operator on $\mathcal{H}$. 
   
 	\begin{definition}
 	Let $T$ be a symmetric operator. 
    \begin{itemize}{}{}
    \item[(1)] We say that $T$ is a {\it time operator} of $H$, 
if there exists a subspace $\mathcal{D}\not =\{0\}$ of $\mathcal{H}$ such that 
$\mathcal{D}\subset D(TH)\cap D(HT)$, and $[T,H]=i$ holds on $\mathcal{D}$.
    \item[(2)] We say that $T$ is a {\it strong time operator} of $H$, 
if the operator equality
\[
e^{itH}Te^{-itH}=T+t
\] 
holds for any $t\in\mathbb{R}$.
    \end{itemize}
 	\end{definition}	
 	
The notion of strong time operator was introduced by Miyamoto \cite{Mi01} in the context of quantum theory.
Their spectra were investigated by Arai \cite{Ar07}.
A purely mathematical study of such pairs $(T,H)$ was done by J{\o}rgensen-Muhly \cite{JM} and Schm\"udgen \cite{Sc83}.
See also \cite{DD}.
Other mathematical studies on strong time operators can be found in 
\cite{Ar05, Ar08, Ar08b, Ar08c, Ar09, Ar10, AH, AM08a, BF, BFH, EM99, HKM09, HKS, RT, We90}.
The paper \cite{Ar08c} is a review of this subject.

The following remarks are fundamental.
In the sequel, we will discuss their analogs for time operators of a unitary operator.

 	\begin{remark}\label{remark sa}
 		Let $T$ be a strong time operator of $H$. 
Then the following hold.
\begin{itemize}{}{}
\item[(1)] Its closure $\bar{T}$ is a strong time operator of $H$ as well.
This follows from a simple limiting argument.
\item[(2)] If $T$ is closed, then $T$ is a time operator of $H$ with a dense subspace $\mathcal{D}=D(TH)\cap D(HT)$ \cite[Proposition 2.1]{Mi01}.
Note that the converse is not true, i.e., not every (closed) time operator of $H$ is a strong time operator of $H$.
We will see such examples in Section \ref{cqw}.
It is worth mentioning that Galapon \cite{Galapon} constructed a time operator of a self-adjoint operator with purely discrete spectrum.
Such a time operator is not a strong time operator 
because every self-adjoint operator admitting a strong time operator has no eigenvalues \cite[Corollary 4.3]{Mi01}.
See also \cite{Ar09, AM}.
\item[(3)] $\sigma (T)=\sigma (T+t)$ holds for any $t\in\mathbb{R}$.
In particular, $T$ is unbounded.
\item[(4)] If $T$ is essentially self-adjoint, then $\sigma(H)=\mathbb{R}$.
For a proof, see e.g., \cite[Theorem 3.2]{Ar08b}.
\end{itemize}
 	\end{remark}

\begin{remark}\label{Te}
For any non-scalar self-adjoint operator, its time operator does exist \cite[Theorem 2.2]{Te}.
\end{remark}

 	\subsection{Time operators of a unitary operator}
We are now ready to give definitions of time operators of a unitary operator $U$.
It will be done by replacing $e^{-itH}$ in the definition of strong time operators of a self-adjoint operator $H$ by $U$.

 	\begin{definition}
Let $U$ be a unitary operator on $\mathcal{H}$, and let $T$ be a symmetric operator. 
 		\begin{itemize}{}{}
    \item[(1)] We say that $T$ is a {\it time operator} of $U$, 
if there exists a subspace $\mathcal{D}\not =\{0\}$ of $\mathcal{H}$ such that 
$\mathcal{D}\subset D(TU)\cap D(T)$, and $[T,U]=U$ holds on $\mathcal{D}$.
    \item[(2)] We say that $T$ is a {\it strong time operator} of $U$, 
if the operator equality
\[
U^*TU=T+1
\] 
holds.
    \end{itemize}
 	\end{definition}

Let us compare fundamental properties of time operators of self-adjoint operators with those of time operators of unitary operators.
 	
 	\begin{remark}\label{equT}
Let $T$ be a strong time operator of a unitary operator $U$.
Then the following hold.
\begin{itemize}{}{}
\item[(1)] Its closure $\bar{T}$ is a strong time operator of $U$ as well.
This follows from a simple limiting argument.
\item[(2)] $T$ is a time operator of $U$ with a dense subspace $\mathcal{D}=D(T)$.
Note that the converse in not true, i.e., not every time operator of $U$ is a strong time operator of $U$ as we will see in Example \ref{hadamard}.
\item[(3)] $\sigma (T)=\sigma (T+1)$ holds.
In particular, $T$ is unbounded.
\end{itemize}
 	\end{remark}

 Let $\mathbb{T}$ be the set of complex numbers of modulus one,
 i.e., 
 $\mathbb{T} = \{ z\in \mathbb{C} \mid |z|=1 \}$.

 	\begin{proposition}\label{prS}
 		Let $T$ be a strong time operator of a unitary operator $U$. 
If $T$ is essentially self-adjoint, then $\sigma (U)=\mathbb{T}$.
 	\end{proposition}
 	
 	\begin{proof}
 		We denote by $\bar{T}$ the closure of $T$. 
Then we can show that 
\[
U^*e^{it\bar{T}}U = e^{it(\bar{T}+1)}= e^{it}e^{it\bar{T}}
\]
for any $t\in\mathbb{R}$.
 		Thus, we obtain $e^{it\bar{T}}Ue^{-it\bar{T}}=e^{it}U$. 
This means that $\sigma (U)=\sigma (e^{it}U)$ for all $t\in \mathbb{R}$. 
Hence, $\sigma (U)=\mathbb{T}$ holds.
 	\end{proof}
 
The following is a toy example, but it makes a difference between time operators 
of self-adjoint operators and those of unitary operators stand out clearly.

\begin{example}\label{toy}
Let $\mathcal{H}=\ell^2(\mathbb{Z})$.
We define a unitary operator $L$ and a self-adjoint operator $X$ by
\[
D(L):=\mathcal{H},\ \ \ \ \ (L\psi)(n):=\psi(n+1),\ \ \ \ \ \psi\in\mathcal{H},\ \ n\in \mathbb{Z}
\]
and
\[
D(X) := \left\{\psi\in\mathcal{H} \ \Big|\ \sum_{n\in\mathbb{Z}}|n\psi(n)|^2<\infty\right\},\ \ \ 
(X\psi)(n):=n\psi(n),\ \ \  \psi\in D(X),\ \ n\in \mathbb{Z}.
\]
A direct calculation shows that $-X$ is a strong time operator of $L$. 
Since $X$ is self-adjoint, we have $\sigma (L)=\mathbb{T}$
by Proposition \ref{prS}.

Note that the spectrum of $X$ coincides with the set of integers, each element of which is an eigenvalue of $X$ as well.
In particular, $X$ does have an eigenvalue.
This phenomena is in contrast to the case of time operators of self-adjoint operators, 
because strong time operators of self-adjoint operators have no eigenvalues \cite[Corollary 4.2]{Mi01}.
\end{example}
 	
A strong time operator is related to the dynamics of a quantum system. 
If the time evolution is described by a unitary operator $U$, then the transition probability between states is bounded by the amount 
which is determined only from the time operator of $U$ and states.
More precisely, we have the following.

\begin{theorem}\label{decay}
Let $T$ be a strong time operator of a unitary operator $U$, and let $S$ be a bounded operator commuting with $U$.
Then for any $\psi\in D(T)$, $\phi\in D(T^*)$ and $t\in\mathbb{Z}\setminus\{0\}$,
	\begin{align*}
	|\langle \phi,U^t\psi\rangle |
\leq \frac{\|(T^*+S^*)\phi\|\|\psi\|+\|\phi\|\|(T+S)\psi\|}{|t|}
	\end{align*}
holds.
\end{theorem}

\begin{proof}
The theorem is a unitary analog of \cite[Proposition 3.1]{Ar05}.
We follow the proof of it.
Since we have $U^t(T+S) =(T+S-t)U^{t}$ for any $t\in\mathbb{Z}$, we obtain
\begin{align*}
|\langle\phi ,U^t\psi  \rangle|
&=\left|\left\langle\phi ,\frac{(T+S)U^t-U^t(T+S)}{t}\psi\right\rangle\right|\\
&=\left|\frac{1}{t}\langle (T^*+S^*)\phi ,U^t\psi\rangle-\frac{1}{t}\langle U^{-t}\phi ,(T+S)\psi\rangle\right|\\
&\leq \frac{\|(T^*+S^*)\phi\|\|\psi\|+\|\phi\|\|(T+S)\psi\|}{|t|},
\end{align*}
which is the desired result.
\end{proof}

For an operator $A$ and a unit vector $\psi\in D(A)$, 
we define the {\it uncertainty} $(\Delta A)_{\psi}$ by
\[
(\Delta A)_{\psi} :=\|A\psi-\langle\psi,A\psi\rangle\psi\|.
\]
The following corollary estimates the survival probability at time $t$.

\begin{corollary}\label{bound by uncertainly}
Let $T$ be a strong time operator of a unitary operator $U$, and let $\psi\in D(T)$ be a unit vector.
Then we have
\[
|\langle\psi, U^t\psi\rangle|^2 \leq  \frac{4(\Delta T)_{\psi}^2}{t^2}
\]
for any $t\in\mathbb{Z}\setminus\{0\}$. 
\end{corollary}

\begin{proof}
This follows from Theorem \ref{decay} with $S=-\langle\psi,T\psi\rangle$ and $\phi=\psi$.
\end{proof}

\begin{remark}
Corollary \ref{bound by uncertainly} is a unitary analog of \cite[Theorem 4.1]{Mi01}.
For the proof, we have followed the argument in \cite[Remark 8.3]{Ar05}.
\end{remark}

\begin{corollary}\label{no eigenvalues}
Let $U$ be a unitary operator admitting a strong time operator $T$.
Then $U$ has no eigenvalues.
\end{corollary}

\begin{proof}
The corollary is a unitary analog of \cite[Corollary 4.3]{Mi01}.
We follow the proof of it.
Assume that $U$ has an eigenvalue $\lambda$, and we shall derive a contradiction.
Let $\psi$ be a normalized eigenvector corresponding to $\lambda$.
Since $D(T)$ is dense, there exists a sequence $\psi_n\in D(T)$ such that $\psi_n\to\psi$ as $n\to\infty$.
Then by Corollary \ref{bound by uncertainly}, we have
\begin{align*}
1 &= |\langle\psi,U^t\psi\rangle| \leq \|\psi-\psi_n\| + \|\psi_n\|\|\psi-\psi_n\| + |\langle\psi_n,U^t\psi_n\rangle|\\
&\leq \|\psi-\psi_n\| + \|\psi_n\|\|\psi-\psi_n\| + \|\psi_n\|^2\cdot\frac{2(\Delta T)_{\psi_n/\|\psi_n\|}}{|t|}
\end{align*}
for any $n,t\in\mathbb{N}$.
After letting $t\to\infty$, we also let $n\to\infty$ and we get $1\leq 0$, which is a contradiction.
Hence the corollary follows.
\end{proof}

\begin{theorem}
Let $T$ be a strong time operator of a unitary operator $U$.
Then for any  $n\in\mathbb{N}$, $\psi\in D(T^n)$ and $\phi\in D\bigl((T^*)^n\bigr)$, there exists a constant $C_n(\phi,\psi)>0$ such that
	\begin{align*}
	|\langle \phi,U^t\psi\rangle |
\leq \frac{C_n(\phi,\psi)}{|t|^n},\ \ \ \ \ t\in\mathbb{Z}\setminus\{0\}
	\end{align*}
holds.
\end{theorem}

\begin{proof}
The theorem is a unitary analog of \cite[Theorem 8.5]{Ar05}.
We follow the proof of it.
We prove the theorem by induction on $n$.
The case $n=1$ follows from Theorem \ref{decay}.
Assume that the theorem holds for any $n=1,\cdots,m-1$ with $m\geq 2$.
Take any $\psi\in D(T^m)$, $\phi\in D\bigl((T^*)^m\bigr)$ and $t\in\mathbb{Z}\setminus\{0\}$.
Since $U^{-t}TU^t=T+t$, we have 
\[
U^{-t}T^mU^t = (T+t)^m = \sum_{j=0}^{m}\binom{m}{j}T^jt^{m-j} = t^{m}+\sum_{j=1}^{m}\binom{m}{j}T^jt^{m-j},
\]
whence
\[
T^mU^t-U^{t}\sum_{j=1}^{m}\binom{m}{j}T^jt^{m-j} = U^tt^{m}
\]
holds.
Keeping in mind that $T^j\psi\in D(T^{m-j})$ for $j=1,\cdots,m-1$, we get
\begin{align*}
|t|^m|\langle\phi,U^t\psi\rangle| &= \left|\langle\phi,T^mU^t\psi\rangle-\left\langle\phi,U^{t}\sum_{j=1}^{m}\binom{m}{j}T^jt^{m-j}\psi\right\rangle\right|\\
&\leq \left\|\left(T^{*}\right)^m\phi\right\|\|\psi\| + \|\phi\|\|T^m\psi\| + \sum_{j=1}^{m-1}\binom{m}{j}|t|^{m-j}|\langle\phi,U^{t}T^j\psi\rangle|\\
&\leq \left\|\left(T^{*}\right)^m\phi\right\|\|\psi\| + \|\phi\|\|T^m\psi\| + \sum_{j=1}^{m-1}\binom{m}{j}\cdot C_{m-j}(\phi,T^{j}\psi).
\end{align*}
The right-hand side is independent of $t$, hence the theorem with $n=m$ follows.
This finishes the proof.
\end{proof}

We close this section with the following theorem, which guarantees the existence of a time operator.

 \begin{theorem}\label{th:te}
 	Let $U$ be a non-scalar unitary operator on $\mathcal{H}$. 
Then there exists a time operator $T$ of $U$.
 \end{theorem}
 
 \begin{proof}
We split the proof into two cases.

{\bf Case 1.} $U$ has at least two distinct eigenvalues.\\
In this case, we modify the proof of \cite[Lemma 2.1]{Te} to construct a bounded self-adjoint time operator of $U$.
Let $\lambda$, $\mu$ be two distinct eigenvalues of $U$, and let $\xi$, $\eta$ be normalized eigenvectors corresponding to $\lambda$, $\mu$, respectively.
Note that $\xi$ and $\eta$ are mutually orthogonal.
Define a bounded self-adjoint operator $T$ on $\mathcal{H}$ by
\[
T\xi := \frac{i\mu}{\lambda-\mu}\eta,\ \ \ \ \ T\eta := \frac{i\lambda}{\lambda-\mu}\xi,\ \ \ \ \ T\psi := 0,\ \ \psi\in{\rm span}\{\xi,\eta\}^{\perp}.
\]
Let $\mathcal{D}$ be the one-dimensional subspace spanned by the vector $\xi+i\eta$.
We show that $T$ satisfies $[T,U]=U$ on $\mathcal{D}$.
By direct computations, we see that
\[
TU(\xi+i\eta) = T(\lambda\xi+i\mu\eta) = \frac{i\lambda\mu}{\lambda-\mu}\eta-\frac{\lambda\mu}{\lambda-\mu}\xi,
\]
and
\[
UT(\xi+i\eta) = U\left(\frac{i\mu}{\lambda-\mu}\eta+i\frac{i\lambda}{\lambda-\mu}\xi\right) = \frac{i\mu^2}{\lambda-\mu}\eta-\frac{\lambda^2}{\lambda-\mu}\xi,
\]
hence we get
\[
[T,U](\xi+i\eta) = i\mu\eta+\lambda\xi = U(\xi+i\eta).
\]
Therefore $T$ is a time operator of $U$.

{\bf Case 2.} $U$ has at most one eigenvalue.\\
In this case, we can choose a real number $\theta\in\mathbb{R}$ for which $1-e^{i\theta}U$ is injective. 
We define the self-adjoint operator $H$ by the Cayley transformation of $e^{i\theta}U$.
 	Note that we have
 	\begin{equation*}
 	U=e^{-i\theta}(H-i)(H+i)^{-1}.
 	\end{equation*}
 	We now use \cite[Theorem 2.2]{Te} to get a symmetric operator $T'$ of $H$ and a subspace $\mathcal{D}$ of $\mathcal{H}$ 
with $\mathcal{D}\subset D(T'H)\cap D(HT')$, $[T',H]=i$ on $\mathcal{D}$, and $\mathcal{D}\cap{\rm D}(HT'H)\neq\{0\}$. 
The last statement is not explicitly written in the theorem, but it follows from the proof of \cite[Theorem 2.2]{Te}.
Let 
\[
T:=-\frac{1}{2}(H-i)T'(H+i).
\] 
By the definitions of $T$ and $H$, 
we have $U(\mathcal{D}\cap{\rm D}(HT'H))\subset{\rm D}(T)$. 
In the sequel, we have to pay attention to domains of operators.
For all $\psi\in\mathcal{D}\cap{\rm D}(HT'H)$, we see that
 	\begin{align*}
 	[T,U]\psi &=e^{-i\theta}[T,e^{i\theta}U]\psi
= e^{-i\theta}[T,(H-i)(H+i)^{-1}]\psi\\
 	&=\frac{-e^{-i\theta}}{2}\Bigl((H-i)T'(H+i)(H-i)(H+i)^{-1}\\
&\ \ \ \ \ \ \ \ \ \ \ \ \ \ \ \ -(H-i)(H+i)^{-1}(H-i)T'(H+i)\Bigr)\psi\\
    &=\frac{-e^{-i\theta}}{2}\Bigl((H-i)T'(H-i)
 -(H-i)(H+i)^{-1}(H-i)T'(H+i)\Bigr)\psi\\
    &=\frac{-e^{-i\theta}}{2}(H-i)(H+i)^{-1}\Bigl((H+i)T'(H-i)-(H-i)T'(H+i)\Bigr)\psi\\
    &=\frac{-e^{-i\theta}}{2}(H-i)(H+i)^{-1}(-2iHT'+2iT'H)\psi\\
    &=-ie^{-i\theta}(H-i)(H+i)^{-1}[T',H]\psi
=e^{-i\theta}(H-i)(H+i)^{-1}\psi\\
 	&=U\psi.
 	\end{align*}
 	Hence $T$ is a time operator of $U$. 
This completes the proof.
 \end{proof}

\section{Time operators for continuous-time quantum walks\label{cqw}}

In this section, we consider a continuous-time quantum walk on the lattice $\mathbb{Z}^d$ with $d\in\mathbb{N}$.
For such a model, we first construct a strong time operator.
Then we compute their deficiency indices and spectra.

\subsection{Construction of time operators}\label{construction continuous}

The Hilbert space of the system is $\mathcal{H}_d:=\ell^2(\mathbb{Z}^d)$, and the Hamiltonian is defined by
\[
(H_d\psi)(x) := \frac{1}{2d}\sum_{|y-x|=1}\psi(y),\ \ \ \ \ \psi\in\mathcal{H}_d,\ \ x\in\mathbb{Z}^d.
\]
We shall find a strong time operator of $H_d$.
Let $\mathcal{K}_d:=L^2([0,2\pi]^d,dk/(2\pi)^d)$ be the Hilbert space of square-integrable functions on $[0,2\pi]^d$.
We define the Fourier transform $\mathscr{F}_d:\mathcal{H}_d\to\mathcal{K}_d$ as a unitary operator so that 
\[
(\mathscr{F}_d\psi)(k) := \sum_{x\in\mathbb{Z}^d}\psi(x)e^{-ix\cdot k}, \ \ \ \ \ k\in[0,2\pi]^d
\]
for $\psi\in\mathcal{H}_d$ with finite support.
From now on, the multiplication operator on $\mathcal{K}_d$ associated to a function $g:[0,2\pi]^d\to\mathbb{R}$ will be denoted by $g(k)$.
Then we can diagonalize $H_d$ as
\[
\hat{H}_d := \mathscr{F}_d H_d\mathscr{F}_d^{-1} = \frac{1}{d}\sum_{j=1}^d\cos{k_j},
\]
where $k=(k_1,\cdots,k_d)\in[0,2\pi]^d$.
In particular, $\sigma(H_d)=[-1,1]$ holds.
Thus, to construct a strong time operator of $H_d$, it is sufficient to find strong time operators of multiplication operators on $\mathcal{K}_d$.\\

We begin with the case $d=1$. 
Let $AC[0,2\pi]$ be the set of absolutely continuous functions on the interval $[0,2\pi]$.
Define an operator $P$ on $\mathcal{K}_1=L^2([0,2\pi],dk/2\pi)$ by $Pf:=-if'$ with domain 
\[
D(P) := \left\{f\in AC[0,2\pi] \mid f'\in \mathcal{K}_1,\ f(0)=f(2\pi)\right\}.
\]
It is known that $P$ is a self-adjoint operator.
We use the following theorem.

\begin{theorem}\label{lemma sa}
Let $g:\mathbb{R}\to\mathbb{R}$ be a twice continuously differentiable periodic function with period $2\pi$ whose derivative $g'$ has at most finitely many zeros in $[0,2\pi)$.
Then the operator 
\[
T := -\frac{1}{2}\left(\frac{1}{g'(k)}P+P\frac{1}{g'(k)}\right)
\]
on $\mathcal{K}_1$ with natural domain is a strong time operator of the self-adjoint operator $g(k)$.
\end{theorem}

\begin{proof}
We first show that $T$ is a symmetric operator.
Let $Z$ be the set of zeros of $g'$ in $(0,2\pi)$, which is a finite set by assumption.
Let $\Omega:=(0,2\pi)\setminus Z$.
Then $C_0^{\infty}(\Omega)$ is dense in $\mathcal{K}_1$, thus the domain of $T$ is dense.
On the other hand, we have
\[
T^* \supset -\frac{1}{2}\left(P\frac{1}{g'(k)}+\frac{1}{g'(k)}P\right) = T,
\] 
which means that $T$ is symmetric.

Next we show that $T$ is a strong time operator of $g(k)$.
Let $t\in\mathbb{R}$ be arbitrary.
Since $D(P)$ is invariant under $e^{itg(k)}$, so is $D(T)$.
Thus for any $f\in D(T)$, we have $e^{itg(k)}f\in D(T)$, and we see that
\[
e^{itg(k)}Tf = e^{itg(k)}\cdot\frac{i}{2}\left\{\frac{1}{g'(k)}f'+\frac{d}{dk}\left(\frac{1}{g'(k)}f\right)\right\}
\]
and
\begin{align*}
Te^{itg(k)}f &= \frac{i}{2}\left\{ite^{itg(k)} f+\frac{1}{g'(k)}e^{itg(k)}f'
+ite^{itg(k)} f+e^{itg(k)}\cdot\frac{d}{dk}\left(\frac{1}{g'(k)}f\right)\right\}\\
&= e^{itg(k)}\cdot\frac{i}{2}\left\{\frac{1}{g'(k)}f'+\frac{d}{dk}\left(\frac{1}{g'(k)}f\right)\right\} -te^{itg(k)}f\\
&= e^{itg(k)}Tf -te^{itg(k)}f,
\end{align*}
whence we get
\[
e^{itg(k)}Tf = (T+t)e^{itg(k)}f.
\]
Summing up the above arguments, we obtain the inclusion $e^{itg(k)}T\subset  (T+t)e^{itg(k)}$.
The converse inclusion immediately follows by replacing $t$ by $-t$.
This finishes the proof.
\end{proof}

\begin{remark}
Theorem \ref{lemma sa} is an analog of \cite[Theorem 1.9]{HKM09}.
See also \cite[Theorem 2.4]{AM08a}.
\end{remark}

By Theorem \ref{lemma sa}, we get a strong time operator
\[
\hat{T}_1:=\frac{1}{2}\left(\frac{1}{\sin{k}}P+P\frac{1}{\sin{k}}\right)
\]
of $\hat{H}_1=\cos{k}$.
To compute the inverse Fourier transform of $\hat{T}_1$, let $X, L$ be the operators on $\mathcal{H}_1$ defined in Example \ref{toy}.
Direct calculations show that we have $\mathscr{F}_1^{-1}P\mathscr{F}_1 = -X$ and $\mathscr{F}_1^{-1}e^{ik}\mathscr{F}_1 = L$.
Hence we conclude that
\[
T_1 := \mathscr{F}_1^{-1}\hat{T}_1\mathscr{F}_1 = -\frac{1}{2}\Bigl({\rm Im}(L)^{-1}X+X\,{\rm Im}(L)^{-1}\Bigr)
\]
is a strong time operator of $H_1={\rm Re}(L)$, where
\[
{\rm Re}(L) := \frac{L+L^*}{2},\ \ \ \ \ {\rm Im}(L) := \frac{L-L^*}{2i}.
\]

\subsection{Deficiency indices and spectra of time operators}

We next compute the deficiency indices and the spectrum of $T_1$.
Recall that the {\it deficiency indices} of a symmetric operator $A$ are defined by
\begin{align*}
d_+(A)&:=\dim\ker{(A^*-i)} = \dim\ker{(A^*-\lambda)},\ \ \ \ \  {\rm Im}(\lambda)>0,\\
d_-(A)&:=\dim\ker{(A^*+i)} = \dim\ker{(A^*-\lambda)},\ \ \ \ \  {\rm Im}(\lambda)<0.
\end{align*}

The following example is fundamental.

\begin{example}\label{reduced ex}
Let $I:=[a,b]$ be a compact interval, and let $dE$ be the Lebesgue measure on $I$.
We define an operator $P_I$ on $L^2(I,dE)$ by
\[
D(P_I) := \{f\in AC[a,b] \mid f'\in L^2(I,dE),\ f(a)=f(b)=0\}, \ \ \ \ \ P_If:=-if'.
\]
Then a direct computation shows that $-P_I$ is a strong time operator of the multiplication operator by the variable $E\in I$.
It is known that the subspace ${\rm ker}(-P_I^*\mp i)$ is the one-dimensional subspace spanned by $e^{\pm E}$.
Hence, $d_{\pm}(-P_I)=1$ holds. 
For the details, see e.g., \cite[Example 3.2]{Sc12}.
This, in particular, implies that $-P_I\mp z$ is not surjective for any $z\in\mathbb{C}\setminus\mathbb{R}$, whence $\sigma(-P_I)=\mathbb{C}$.
\end{example}

We reduce calculations of the deficiency indices and the spectrum of $\hat{T}_1$ to those of the above example, by passing to the energy representation.
Egusquiza and Muga  \cite{EM99} first used the energy representation to determine the deficiency indices of Aharonov and Bohm's time operator.
It is in fact a structure theorem for those time operators defined in Theorem \ref{lemma sa}.

\begin{theorem}\label{deficiency indices sa}
Let $g$ and $T$ be as stated in Theorem \ref{lemma sa}.
Then the following statements hold.
\begin{itemize}{}{}
\item[(1)] There exists a unitary operator $V:\mathcal{K}_1\to\oplus_{j=1}^nL^2\left(g(\bar{I}_j),dE\right)$ such that
\[
V\bar{T}V^{-1}
 = \bigoplus_{j=1}^n(-P_j),
\ \ \ \ \ Vg(k)V^{-1} = \bigoplus_{j=1}^n E,
\]
where  $0\leq a_1<a_2<\cdots<a_n<2\pi$ are the zeros of $g'$ in $[0,2\pi)$, $a_{n+1}:=2\pi+a_1$, and
$I_j:=(a_j,a_{j+1})$ is an open interval
, and $P_j$ is the operator $P_I$ defined in Example \ref{reduced ex} 
with $I$ replaced by the closed interval $g(\bar{I}_j)$ 
for $j=1,\cdots,n$.

\item[(2)] The deficiency indices of $T$ are equal to the number of zeros of $g'$ in $[0,2\pi)$.
In particular, $T$ has a self-adjoint extension.
\end{itemize}
\end{theorem}

Before proving Theorem \ref{deficiency indices sa}, 
we state the following corollary. 

\begin{corollary}\label{spectrum sa}
Let $g$ and $T$ be as stated in Theorem \ref{lemma sa}. 
Then $\sigma(T) = \mathbb{C}$. 
\end{corollary}

\begin{proof}
This follows from Theorem \ref{deficiency indices sa} (1) and Example \ref{reduced ex}.
Note that the corollary also follows from \cite[Theorem 2.1 (iii)]{Ar07}.
\end{proof}

\begin{proof}[Proof of Theorem \ref{deficiency indices sa}]
(1) We first note that since $g(0)=g(2\pi)$, Rolle's theorem implies that there exists at least one zero of $g'$ in $(0,2\pi)$, thus $n\geq 1$.
Also, remark that $g(\bar{I}_j)$ is a compact interval with endpoints $g(a_j)$ and $g(a_{j+1})$.
Let $g_j:=g|_{\bar{I}_j}$.
We regard each function in $\mathcal{K}_1$ as a periodic function with period $2\pi$ defined on the whole $\mathbb{R}$.
Define the linear operator $V_j:\mathcal{K}_1\to L^2(g(\bar{I}_j),dE)$ by
\[
(V_jf)(E) := \frac{f(g_j^{-1}(E))}{\sqrt{2\pi |g'(g_j^{-1}(E))|}},\ \ \ \ \ f\in\mathcal{K}_1,\ \ E\in g(I_j).
\]
Since
\[
\|V_jf\|^2 = \int_{a_j}^{a_{j+1}}|f(k)|^2\,\frac{dk}{2\pi},\ \ \ \ \ f\in\mathcal{K}_1,\
\]
the map
\[
V:\mathcal{K}_1\to\bigoplus_{j=1}^nL^2(g(\bar{I}_j),dE),\ \ \ \ \ f\mapsto \{V_jf\}_{j=1}^n
\]
is a unitary operator.

In order to show $V_jf\in D(P_j)$ for a fixed function $f\in D(T)$, we have only to check the following three conditions:
\begin{itemize}{}{}
\item[(i)] $V_jf$ is absolutely continuous on the compact interval $g(\bar{I}_j)$,
\item[(ii)] the derivative of $V_jf$ is in $L^2(g(\bar{I}_j),dE)$,
\item[(iii)] $\lim_{E\to g(a_j)}(V_jf)(E)=\lim_{E\to g(a_{j+1})}(V_jf)(E)=0$.
\end{itemize}
We first show (iii). 
Since $f\in D(T)$, we have $f/g'\in D(P)$, in particular, $f/g'$ is continuous on $[0,2\pi]$.
Thus
\[
(V_jf)(E) = \frac{f(g_j^{-1}(E))}{\sqrt{2\pi |g'(g_j^{-1}(E))|}} = \frac{f(g_j^{-1}(E))}{|g'(g_j^{-1}(E))|}\cdot \sqrt{\frac{|g'(g_j^{-1}(E))|}{2\pi}} \to 0,
\]
as $E$ tends to $g(a_j)$ and $g(a_{j+1})$.
Thus (iii) holds.
We next show (ii).
By a direct computation, we see that
\[
i\frac{d}{dE}(V_jf)(E) = (V_jTf)(E) \in L^2(g(\bar{I}_j),dE) \subset L^1(g(\bar{I}_j),dE),
\]
whence (ii) follows.
We split the proof of (i) into two cases.
If $g'(E)>0$ on $g(I_j)$, then for any $g(a_j)<E<g(a_{j+1})$ and for any $\varepsilon>0$, $V_jf$ is absolutely continuous on the compact interval $[g(a_j)+\varepsilon,E]$, 
thus
\[
\int_{g(a_j)+\varepsilon}^E \frac{d}{dF}(V_jf)(F)\, dF = (V_jf)(E)-(V_jf)(g(a_j)+\varepsilon).
\]
By (iii), we get
\[
\int_{g(a_j)}^E \frac{d}{dF}(V_jf)(F)\,dF = \lim_{\varepsilon\to 0}\int_{g(a_j)+\varepsilon}^E \frac{d}{dF}(V_jf)(F)\, dF = (V_jf)(E),
\]
which also holds for $E=g(a_j), g(a_{j+1})$ by (iii).
Hence $V_jf$ is absolutely continuous on $g(\bar{I}_j)$.
Similarly, one can show that $V_jf$ is absolutely continuous on $g(\bar{I}_j)$ if $g'(E)<0$ on $g(I_j)$.
Therefore $V_jf\in D(P_j)$ follows.
Moreover, in the above proof, we have shown $(-P_j)V_jf=V_jTf$.
Thus $V_jT\subset (-P_j)V_j$, 
whence $VTV^{-1}\subset \oplus_{j=1}^n(-P_j)$ holds.

On the other hand, since $C_0^{\infty}(g(I_j))$ is a core for $-P_j$, $\oplus_{j=1}^n C_0^{\infty}(g(I_j))$ is a core for $\oplus_{j=1}^n(-P_j)$.
Moreover $D(TV^{-1})$ contains $\oplus_{j=1}^n C_0^{\infty}(g(I_j))$.
Hence we conclude that $V\bar{T}V^{-1}=\oplus_{j=1}^n(-P_j)$.
The other operator equality $Vg(k)V^{-1}=\oplus_{j=1}^nE$ immediately follows from the definition of $V$.

(2) Since the deficiency indices of $\oplus_{j=1}^n(-P_j)$ are $n$, we get $d_{\pm}(T)=d_{\pm}(\bar{T})=n$.
This completes the proof.
\end{proof}

By Theorem \ref{deficiency indices sa}, we obtain $d_{\pm}(T_1)=d_{\pm}(\hat{T}_1)=2$.
More specifically, there exists a unitary operator $V:\mathcal{K}_1\to L^2([-1,1],dE)\oplus L^2([-1,1],dE)$ such that
\[
V\overline{\hat{T}_1}V^{-1} = (-P_{[-1,1]})\oplus (-P_{[-1,1]}),\ \ \ \ \ V\hat{H}_1V^{-1} = E\oplus E.
\]
This tells us that the subspace $\ker{\left(\hat{T}_1^*\mp i\right)}$ is the two-dimensional subspace spanned by 
\[
\sqrt{\sin{k}}\,e^{\pm \cos{k}}\chi_{[0,\pi]}(k)\ \ \textrm{and}\ \  \sqrt{-\sin{k}}\,e^{\pm \cos{k}}\chi_{[\pi,2\pi]}(k),
\]
where $\chi_A$ is the characteristic function for a subset $A$ of $\mathbb{R}$.
It follows from Corollary \ref{spectrum sa} that $\sigma(T_1)=\mathbb{C}$.

We remark that all self-adjoint extensions of $T_1$ are time operators of $H_1$, but none of them are strong time operators of $H_1$.
The latter follows from the fact that $\sigma(H_1)=[-1,1]$ and Remark \ref{remark sa} (4).\\

\subsection{The case  $d\geq 2$}
Let us consider the case $d\geq 2$.
We identify $\mathcal{H}_d=\ell^2(\mathbb{Z}^d)$ with $\otimes^d\mathcal{K}_1=\otimes^d\ell^2(\mathbb{Z})$.
Under the identification, we have
\[
H_d = \frac{1}{d}\sum_{j=1}^d 1\otimes \cdots \otimes 1 \otimes \overset{j\,{\rm th}}{H_1} \otimes 1\otimes  \cdots \otimes 1.
\]
Let $\alpha=(\alpha_1,\cdots,\alpha_d)$ be a $d$-tuple of real numbers satisfying $\sum_{j=1}^d\alpha_j=1$.
Then the operator
\[
T_d^{(\alpha)} := d\cdot\sum_{j=1}^d \alpha_j\cdot 1\otimes \cdots \otimes 1 \otimes \overset{j\,{\rm th}}{T_1} \otimes 1\otimes  \cdots \otimes 1.
\]
with domain $\otimes_{\rm alg}^dD(T_1)$ is a strong time operator of $H_d$.
Here, $\otimes_{\rm alg}$ denotes the algebraic tensor product.

To compute the deficiency indices of $T_d^{(\alpha)}$, we may assume that $\alpha_1\not=1$.
Note that $\left(T_d^{(\alpha)}\right)^*$ is an extension of the operator 
\[
d\cdot\sum_{j=1}^d\alpha_j\cdot 1\otimes \cdots \otimes 1 \otimes \overset{j\,{\rm th}}{T_1^*} \otimes 1\otimes  \cdots \otimes 1
\]
with domain $\otimes_{\rm alg}^dD(T_1^*)$.
Hence for any $0<\varepsilon<1/2$, there are vectors $\psi_{\varepsilon}^{(\pm)}, \phi_{\varepsilon}^{(\pm)}\in D(T_1^*)$ 
such that 
\[
T_1^*\psi_{\varepsilon}^{(\pm)}=\pm i\varepsilon\psi_{\varepsilon}^{(\pm)},\ \ \ \ \ 
T_1^*\phi_{\varepsilon}^{(\pm)}=\pm i\frac{1-d\alpha_1\varepsilon}{d\sum_{j=2}^d\alpha_j}\phi_{\varepsilon}^{(\pm)}
\]
hold.
Put $\Psi_{\varepsilon}^{(\pm)}:=\psi_{\varepsilon}^{(\pm)}\otimes\phi_{\varepsilon}^{(\pm)}\otimes\cdots\otimes\phi_{\varepsilon}^{(\pm)}\in\otimes_{\rm alg}^dD(T_1^*)$.
Then $T_d^{(\alpha)}\Psi_{\varepsilon}^{(\pm)}=\pm i\Psi_{\varepsilon}^{(\pm)}$ follows.
Since the set $\{\Psi_{\varepsilon}^{(\pm)} \mid 0<\varepsilon<1/2\}$ is linearly independent, we get $d_{\pm}\left(T_d^{(\alpha)}\right)=\infty$.
By a similar argument to Example \ref{reduced ex}, we obtain $\sigma\left(T_d^{(\alpha)}\right)=\mathbb{C}$.

As is the case for $d=1$, all self-adjoint extensions of $T_d^{(\alpha)}$ are time operators of $H_d$, 
but none of them are strong time operators of $H_d$.

\section{Time operators for discrete-time quantum walks}

In this section, we consider discrete-time quantum walks.
We first construct strong time operators for such models.
Then we compute their deficiency indices and spectra.
We also discuss a relation between the strong time operators and the winding numbers.

\subsection{Construction of time operators\label{construction discrete}}

Recall that $\mathbb{T}$ denotes the set of complex numbers of modulus one.
Similarly to Theorem \ref{lemma sa}, we first construct a strong time operator of the unitary operator 
$\lambda(k)$ on $\mathcal{K}_1$,
where $\lambda:\mathbb{R} \to \mathbb{T}$ is a periodic function
with period $2\pi$.  
Note that for any continuous periodic function $\lambda:\mathbb{R}\to \mathbb{T}$ with period $2\pi$, 
there exists a unique integer $m\in\mathbb{Z}$ 
so that 
\begin{equation}
\label{wind}
\lambda(k)=e^{i(mk+\theta(k))}, \quad k\in\mathbb{R}
\end{equation}
with some continuous periodic function $\theta:\mathbb{R}\to\mathbb{R}$ with period $2\pi$.  
The integer $m$ is called the {\it winding number} of $\lambda$.
For a proof, see e.g., \cite[Lemma 3.5.14]{Mu}.

Let $P$ be the operator on $\mathcal{K}_1$ defined in Subsection \ref{construction continuous}.

\begin{theorem}\label{lemma unitary}
Let
$\lambda:\mathbb{R}\to \mathbb{T}$
be a twice continuously differentiable periodic function with period $2\pi$ whose derivative 
$\lambda^\prime$ has at most finitely many zeros in $[0,2\pi)$.
Let $m$ be the winding number of $\lambda$, and let $\theta$ be the function given in \eqref{wind}.  
Then the operator 
\begin{align*}
T  &:= 
\frac{1}{2}
	\left(\frac{1}{m+\theta'(k)}P+P\frac{1}{m+\theta'(k)}\right) 
\end{align*}
on $\mathcal{K}_1$ with natural domain is a strong time operator of the unitary operator $\lambda(k)$. 
\end{theorem}

\begin{proof}
Define a continuous function $g:\mathbb{R} \to \mathbb{R}$ by
\begin{equation}\label{gkm} 
g(k) = m k + \theta(k), \quad k \in \mathbb{R}. 
\end{equation}
Note that $g$ is a twice continuously differentiable function, which follows from the following two facts.
First, differentiability and continuity are both local properties.
Second, the logarithm function is locally well-defined in $\mathbb{C}\setminus\{0\}$.
Because $\lambda(k) = e^{ig(k)}$, Theorem \ref{lemma unitary} can now be proved by the same manner as Theorem \ref{lemma sa}.
\end{proof}

\begin{remark}
Since $\lambda^\prime /\lambda  = i g^\prime$, 
the time operator defined in Theorem \ref{lemma unitary} 
can also be expressed as
\begin{equation}
\label{alternative form} 
 T = \frac{i}{2}\left(\frac{\lambda(k)}{\lambda^\prime(k)}P+P\frac{\lambda(k)}{\lambda^\prime(k)}\right). 
\end{equation}	
\end{remark}

Applying Theorem \ref{lemma unitary},
we construct a strong time operator
of a one-dimensional homogeneous quantum walk.
Let 
$\mathcal{H} =\ell^2(\mathbb{Z};\mathbb{C}^2)$
be the Hilbert space of states for the quantum walker.  
Identifying $\mathcal{H}$ with 
$\ell^2(\mathbb{Z}) \oplus \ell^2(\mathbb{Z})$,
we define a shift operator $S$ on $\mathcal{H}$
as 
$S := \begin{pmatrix}
L&0\\
0&L^*
\end{pmatrix}$,
where $L$ is the unitary operator on $\ell^2(\mathbb{Z})$ defined in Example \ref{toy}. 
Taking a $2\times 2$ unitary matrix 
$C=\begin{pmatrix}
a&b\\
c&d
\end{pmatrix}$,
we define a coin operator on $\mathcal{H}$
as the multiplication operator by $C$,
and we denote it by the same symbol. 
Then the time evolution operator of the quantum walk is defined by $U=SC$.

Let $\mathcal{K}$ be the Hilbert space of square integrable functions $f:[0,2\pi]\to\mathbb{C}^2$ with norm
\[
\|f\| := \left(\int_0^{2\pi}\|f(k)\|_{\mathbb{C}^2}^2\,\frac{dk}{2\pi}\right)^{\frac{1}{2}}.
\]
We define the Fourier transform $\mathscr{F}:\mathcal{H}\to\mathcal{K}$ as a unitary operator so that 
\[
(\mathscr{F}\psi)(k) := \sum_{x\in\mathbb{Z}}e^{-ikx}\psi(x),\ \ \ \ \ k\in[0,2\pi]
\]
for $\psi\in\mathcal{H}$ with finite support.
Then the Fourier transform $\mathscr{F}U\mathscr{F}^{-1}$ of $U$
is the multiplication operator on $\mathcal{K}$ by 
the 2$\times$2 unitary matrix  
\[
\hat{U}(k):=\begin{pmatrix}
e^{ik}a&e^{ik}b\\
e^{-ik}c&e^{-ik}d
\end{pmatrix},\ \ \ \ \ k\in [0,2\pi].
\]
For each $k\in\mathbb{R}$, the unitary matrix $\hat{U}(k)$ has exactly two eigenvalues $\lambda_1(k), \lambda_2(k)$.
Let $v_1(k), v_2(k)$ be the corresponding normalized mutually orthogonal eigenvectors.
By a direct computation (see e.g., Lemma \ref{lemma ev} below), we may assume that the map
$\mathbb{R}\ni k\mapsto \lambda_j(k)\in \mathbb{C}$
is an analytic periodic function with period $2\pi$ for each $j=1,2$.
Similarly, we may assume that the map
$\mathbb{R}\ni k\mapsto v_j(k)\in \mathbb{C}^2$
is a Borel periodic function with period $2\pi$ for each $j=1,2$.
Define a $2\times 2$ unitary matrix $W(k)$ by $W(k):=\bigl(v_1(k),v_2(k)\bigr)$.
Then we obtain
\[
W(k)^{-1}\hat{U}(k)W(k) = 
\begin{pmatrix}
\lambda_1(k)&0\\
0&\lambda_2(k)
\end{pmatrix},\ \ \ \ \ k\in[0,2\pi].
\]
To construct a strong time operator of $U$, it is sufficient to find a strong time operator of the multiplication operator $\lambda_j(k)$ 
on the Hilbert space $\mathcal{K}_1= L^2([0,2\pi],dk/2\pi)$,
because $\mathcal{K} = \mathcal{K}_1 \oplus \mathcal{K}_1$.

As will be seen in Lemma \ref{lemma ev} below,
$\lambda_j(k)$ are constants if and only if $a=0$. 
If $0 < |a| \leq 1$,  
$\lambda_j'(k)$ has at most finitely many zeros in $[0,2\pi)$,
because the map
$
\mathbb{R}\ni k \mapsto \lambda_j'(k)\in\mathbb{C}
$
is a non-constant analytic function. 
Using Theorem \ref{lemma unitary} and the alternative expression \eqref{alternative form},
we can define a strong time operator of $\lambda_j(k)$ by
\begin{equation}
\label{defTj}
\hat{T}_j := \frac{i}{2}\left(\frac{\lambda_j(k)}{\lambda_j'(k)}P+P\frac{\lambda_j(k)}{\lambda_j'(k)}\right)
\end{equation}
for each $j=1,2$.
Let $W$ be the multiplication operator on $\mathcal{K}$ by $W(k)$.
If $0 < |a| \leq 1$, then we can conclude that
\begin{equation}\label{Time op. DTQW}
T :=\mathscr{F}^{-1}W
\begin{pmatrix}
\hat{T}_1&0\\
0&\hat{T}_2
\end{pmatrix}
W^{-1}\mathscr{F}
\end{equation}
is a strong time operator of $U$.

To make the expression \eqref{defTj} explicit,
we shall calculate the winding numbers $m_j$ of $\lambda_j$ and
functions that satisfy \eqref{wind} with $\lambda$, $m$, and $\theta$
replaced by $\lambda_j$, $m_j$, and $\theta_j$. 
To this end, we use $\alpha\in [0,2\pi)$ and $\delta \in [0,2\pi)$
to denote the arguments of $a$ and $\det C$.
Then we can parametrize $C$ as
\[ C = \begin{pmatrix} 
|a|e^{i \alpha} & b \\ - \bar{b}e^{i\delta} & |a|e^{i(\delta-\alpha)} \end{pmatrix}. \]
By direct calculation,
we obtain the following lemma. 

\begin{lemma}
\label{lemma ev}
\begin{itemize}
\item[(1)] If $a = 0$, then 
$\lambda_1(k) \equiv e^{i(\pi+\delta)/2}$
and $\lambda_2(k) \equiv e^{i(-\pi + \delta)/2}$. 
\item[(2)] If $0 < |a| < 1$, then
\[ \lambda_j(k) 
= e^{i \delta/2}\{\tau(k) +i(-1)^{j-1} \sqrt{1-\tau(k)^2} \},
\quad j =1,2, \]
where $\tau(k) := |a| \cos (k + \alpha -\delta/2)$. 
\item[(3)] If $|a|=1$, then
$\lambda_1(k) = e^{i(k+\alpha)} $
and $\lambda_2(k) = e^{i(-k+\delta-\alpha)} $.  
\end{itemize}
\end{lemma}

\begin{remark}\label{conclusion from lemma}
Lemma \ref{lemma ev} tells us the following consequences.
\begin{itemize}
\item[(1)] If $a = 0$, then $\sigma(U)$ consists of two distinct eigenvalues $e^{i(\pi+\delta)/2}$ and $e^{i(-\pi + \delta)/2}$.
Hence, $U$ has no strong time operators, but $U$ has a bounded self-adjoint time operator. 
The former follows from Corollary \ref{no eigenvalues}.
The latter follows from Case 1 in the proof of Theorem \ref{th:te}.
\item[(2)] If $0 < |a| < 1$, then $\sigma(U)\not=\mathbb{T}$.
In particular, $U$ does not have a self-adjoint strong time operator.
This follows from Proposition \ref{prS}.
\end{itemize}
\end{remark}

\begin{lemma}
\label{lemma wind}
\begin{itemize}
\item[(1)] If $a=0$, 
then $m_1 = m_2 = 0$, 
and $\theta_1^\prime = \theta_2^\prime \equiv 0$. 
\item[(2)] If $0 <|a| < 1$,
then $m_1 = m_2 = 0$ and 
\begin{equation} 
\label{thetaprime}
\theta_1^\prime(k) = \frac{d}{dk} \arccos \tau(k),
\quad \theta^\prime_2(k) = - \theta_1^\prime(k), \quad k \in [0,2\pi).
\end{equation}   
\item[(3)] If $|a|=1$, then $m_1 = 1$, $m_2 = -1$
and $\theta_1^\prime = \theta_2^\prime \equiv 0$.
\end{itemize}
\end{lemma}

\begin{proof}
By definition, we obtain
\begin{equation}
\label{windtheta}
m_j + \theta_j^\prime = -i \lambda_j^\prime/\lambda_j.
\end{equation} 
We first consider the case of $a=0$. 
By Lemma \ref{lemma ev} (1), the winding numbers are zero.
Because $\lambda_j^\prime/\lambda_j \equiv 0$,
we observe from \eqref{windtheta} that $\theta_j^\prime \equiv 0$. 
This completes (1).   

We next consider the case of $0 < |a| < 1$. 
By Lemma \ref{lemma ev} (2),
the winding number is calculated by
\begin{align*} 
m_j &= \frac{1}{2\pi i}
	\int_0^{2\pi}\frac{\lambda_j^\prime(k)}{\lambda_j(k)}\,dk 
	= \frac{(-1)^{j-1}}{2\pi}
	\int_0^{2\pi}\frac{-\tau^\prime(k)}
							{\sqrt{1-\tau(k)^2}}\,dk. \\
&= \frac{(-1)^{j-1}}{2\pi}\int_0^{2\pi}\frac{d}{dk} \arccos \tau(k)\,dk\\
&= 0.
\end{align*}
Combining this with \eqref{windtheta},
we have $\theta_2^\prime(k) = - \theta_1^\prime(k)$ and
\[ \theta_1^\prime(k) =
 - \frac{\tau^\prime(k)}
							{\sqrt{1-\tau(k)^2}}
= \frac{d}{dk} \arccos \tau(k), \]
which completes \eqref{thetaprime}.  

Finally, we consider the case of $|a|=1$. 
Lemma \ref{lemma ev} (3) implies that $m_1 = 1$ and that $m_2=-1$. 
Because $-i \lambda_j^\prime/\lambda_j \equiv (-1)^{j-1}=m_j$,
(3) is completed by \eqref{windtheta}.
\end{proof}

As mentioned, we can define $\hat T_j$ by \eqref{defTj}
if $0 < |a| \leq 1$. 
By Lemma \ref{lemma wind}, we obtain the following. 

\begin{theorem}
\label{TimeopofQW}
Let $\hat T_1$ and $\hat T_2$ be strong time operators defined in \eqref{defTj}. 
\begin{itemize}
\item[(1)] If $0 <|a| < 1$,
then 
\[ 
\hat T_1 = \frac{1}{2}
\left( \frac{1}{\displaystyle \frac{d}{dk}  \arccos \tau(k)} P
		+ P  \frac{1}{\displaystyle \frac{d}{dk} \arccos \tau(k)} 
		\right),
		\quad \hat T_2 = -  \hat T_1.
\]							
\item[(2)] If $|a|=1$, then $\hat T_1 = P$ and $\hat T_2 = - \hat T_1$.  
\end{itemize}
\end{theorem}

\subsection{Deficiency indices and spectra of time operators}
We shall compute the deficiency indices and the spectra of $T$ defined in Theorem \ref{lemma unitary}.
The following structure theorem is an analog of Theorem \ref{deficiency indices sa}, 
but the proof is a little bit more complicated because of the winding number of a function.

\begin{theorem}\label{deficiency indices unitary}
Let $\lambda$, $m$, $\theta$, $T$ be as stated in Theorem \ref{lemma unitary}, and define $g$ by \eqref{gkm}. 
\begin{itemize}{}{}
\item[(1)] If $\lambda^\prime$ has no zeros, then $m\not=0$ and $T$ is self-adjoint. 
Moreover, there exists a unitary operator $V:\mathcal{K}_1\to\oplus_{j=0}^{|m|-1}\mathcal{K}_1$ such that
\[
VTV^{-1} = \bigoplus_{j=0}^{|m|-1}\left(P+\frac{j}{|m|}\right),\ \ \ \ \ 
V\lambda(k)V^{-1} = \bigoplus_{j=0}^{|m|-1} e^{ik}.
\]
\item[(2)] If $\lambda^\prime$ has a zero, then there exists a unitary operator 
$V:\mathcal{K}_1\to\oplus_{j=1}^nL^2\left(g(\bar{I}_j),dE\right)$ such that
\[
V\bar{T}V^{-1} = \bigoplus_{j=1}^nP_j,\ \ \ \ \ 
V\lambda(k)V^{-1} = \bigoplus_{j=1}^n e^{iE},
\]
where  $0\leq a_1<a_2<\cdots<a_n<2\pi$ are the zeros of $\lambda^\prime$
in $[0,2\pi)$, $a_{n+1}:=2\pi+a_1$, 
$I_j:=(a_j,a_{j+1})$ is an open interval,
and $P_j$ is the operator $P_I$ defined in Example \ref{reduced ex}
with $I$ replaced by the closed interval $g(\bar{I}_j)$  for $j=1,\cdots,n$.
\item[(3)] The deficiency indices of $T$ are equal to the number of zeros of $\lambda^\prime$ in $[0,2\pi)$.
In particular, $T$ has a self-adjoint extension.
\end{itemize}
\end{theorem}

Before proving Theorem \ref{deficiency indices unitary},
we state the following corollary. 

\begin{corollary}
\label{corspec}
Let $\lambda$, $m$, $T$ be as stated in Theorem \ref{lemma unitary}.
\begin{itemize}
\item[(1)] If $\lambda^\prime$ has no zeros, 
then $\sigma(T) = \left\{s/m \mid s\in\mathbb{Z}\right\}$. 
\item[(2)] If $\lambda^\prime$ has a zero, 
then $\sigma(T) = \mathbb{C}$. 
\end{itemize}
\end{corollary}

\begin{proof}
We prove (1). By Theorem \ref{deficiency indices unitary} (1),
\[ \sigma(T) 
	= \bigcup_{j=0}^{|m|-1} \sigma\left(P + \frac{j}{|m|}\right)
	= \bigcup_{j=0}^{|m|-1} 
		\left\{r + \frac{j}{|m|} \ \Big| \ r \in \mathbb{Z} \right\}
	= \left\{ \frac{s}{m} \ \Big| \  s\in \mathbb{Z} \right\},
\]
where we have used the fact that $\sigma(P) = \mathbb{Z}$. 
(2) follows from Theorem \ref{deficiency indices unitary} (2) and Example \ref{reduced ex}.
\end{proof}

\begin{proof}[Proof of Theorem \ref{deficiency indices unitary}]
(1) If $m=0$, then Rolle's theorem implies that $\theta^\prime$ (and hence $\lambda^\prime$) has at least one zero, which is a contradiction.
Thus $m\not=0$.
We next show that $T$ is self-adjoint.
Since $g'=-i\lambda'/\lambda$, functions $g^\prime$ and $1/g^\prime$
are continuously differentiable periodic functions with period $2\pi$, whence we have
\[
P\frac{1}{g^\prime(k)}
=\frac{1}{g^\prime(k)}P
	+\frac{ig^{\prime\prime}(k)}{g^\prime(k)^2}
\]
as an operator equality.
Thus we obtain
\[
T^* = \left(\frac{1}{g^\prime(k)}P
	+\frac{ig^{\prime\prime}(k)}{2g^\prime(k)^2} \right)^*
= P\frac{1}{g^\prime(k)}
	-\frac{ig^{\prime\prime}(k)}{2g^\prime(k)^2} 
= \frac{1}{g^\prime(k)}P
	+\frac{ig^{\prime\prime}(k)}{2g^\prime(k)^2} = T,
\]
which means that $T$ is self-adjoint.
To see the unitary equivalence, let us define the unitary operator $\tilde{V}:\mathcal{K}_1\to\mathcal{K}_1$ by
\[
(\tilde{V}f)(k) := f\Bigl(g^{-1}(mk+\theta(0))\Bigr)\sqrt{\left|\frac{m}{g'\Bigl(g^{-1}(mk+\theta(0))\Bigr)}\right|}, \ \ \ \ \ f\in\mathcal{K}_1,\ \ k\in [0,2\pi].
\] 
Then the following can be proved in a similar manner to the proof of Theorem \ref{deficiency indices sa} (1):
\[
\tilde{V}T\tilde{V}^{-1} \subset \frac{1}{m}P,\ \ \ \ \ \tilde{V}\lambda(k)\tilde{V}^{-1} = e^{i(mk+\theta(0))}.
\]
Since $\tilde{V}T\tilde{V}^{-1}$ is self-adjoint, it has no proper symmetric extensions, thus the above inclusion must be an equality.
Moreover we see that
\[
e^{-i\theta(0)P/m}\left(\frac{1}{m}P\right)e^{i\theta(0)P/m} = \frac{1}{m}P,\ \ \ \ \ e^{-i\theta(0)P/m}e^{i(mk+\theta(0))}e^{i\theta(0)P/m} = e^{imk}.
\]
For each $j=0,1,\cdots,|m|-1$, we denote by $\mathcal{H}_j$ the closed subspace of $\mathcal{K}_1$ spanned by $\{e^{i(ma+j)k}\in\mathcal{K}_1\mid a\in\mathbb{Z}\}$.
Note that we have the orthogonal decomposition $\mathcal{K}_1=\oplus_{j=0}^{|m|-1}\mathcal{H}_j$.
Since $\mathcal{H}_j$ reduces $P/m$ and $e^{imk}$, the unitary operator $V_j:\mathcal{H}_j\to\mathcal{K}_1$ sending $e^{i(ma+j)k}$ to $e^{iak}$ gives us the following equalities:
\[
V_j\left(\frac{1}{m}P\right)\Big|_{\mathcal{H}_j}V_j^{-1} = P+\frac{j}{m},\ \ \ \ \ V_je^{imk}|_{\mathcal{H}_j}V_j^{-1} = e^{ik}.
\]
Hence, if $m>0$, the unitary operator $V:=\left(\oplus_{j=0}^{|m|-1}V_j\right)e^{-i\theta(0)P/m}\tilde{V}$ satisfies
\[
VTV^{-1} = \bigoplus_{j=0}^{|m|-1}\left(P+\frac{j}{|m|}\right),\ \ \ \ \ V\lambda(k)V^{-1} = \bigoplus_{j=0}^{|m|-1} e^{ik}.
\]
If $m<0$, the operator $V:=\left(V_0\oplus\oplus_{j=1}^{|m|-1}e^{-ik}V_j\right)e^{-i\theta(0)P/m}\tilde{V}$ is the desired unitary operator.

(2) Let $g_j:=g|_{\bar{I}_j}$ for any $j=1, \cdots, n$.
We regard each function in $\mathcal{K}_1$ as a periodic function with period $2\pi$ defined on the whole $\mathbb{R}$.
Define the linear operator 
$V_j:\mathcal{K}_1\to L^2\left(g(\bar{I}_j),dE\right)$ by
\[
(V_jf)(E) := \frac{f(g_j^{-1}(E))}{\sqrt{2\pi |g'(g_j^{-1}(E))|}},\ \ \ \ \ f\in\mathcal{K}_1,\ \ E\in g(I_j).
\]
Since
\[
\|V_jf\|^2 = \int_{a_j}^{a_{j+1}}|f(k)|^2\,\frac{dk}{2\pi},\ \ \ \ \ f\in\mathcal{K}_1,\
\]
the map
\[
V:\mathcal{K}_1\to\bigoplus_{j=1}^nL^2\left(g(\bar{I}_j),dE\right),\ \ \ \ \ f\mapsto \{V_jf\}_{j=1}^n
\]
is a unitary operator, which satisfies
\[
V\bar{T}V^{-1} = \bigoplus_{j=1}^nP_j,\ \ \ \ \ V\lambda(k)V^{-1} = \bigoplus_{j=1}^n e^{iE}.
\]
The details are similar to the proof of Theorem \ref{deficiency indices sa} (1) and we omit them.

(3) This follows from (1) and (2).
\end{proof}

By applying Theorem \ref{deficiency indices unitary} and Corollary \ref{corspec}
to the one-dimensional homogeneous quantum walk, we get the following theorem.

\begin{theorem}
Let $U$ be the time evolution operator of the one-dimensional homogeneous quantum walk,
and let $T$ be the strong time operator of $U$ defined in \eqref{Time op. DTQW}.
\begin{itemize}
\item[(1)] If $0 <|a| < 1$, then $d_{\pm}(T)=4$ and $\sigma(T)=\mathbb{C}$.	
In particular, $T$ has self-adjoint extensions.
All of them are time operators of $U$, but none of them are strong time operators of $U$.					
\item[(2)] If $|a|=1$, then $T$ is self-adjoint and $\sigma(T)=\mathbb{Z}$.
In particular, $T$ is a self-adjoint strong time operator of $U$.
\end{itemize}
\end{theorem}

\begin{proof}
(1) Since
\[
|\lambda_j^\prime(k)| = \left|\frac{\lambda_j^\prime(k)}{\lambda_j(k)}\right|
= \frac{|\tau^\prime(k)|}{\sqrt{1-\tau{(k)}^2}},\ \ \ \ \ k\in\mathbb{R},
\]
the number of zeros of $\lambda_j^\prime$ in $[0,2\pi)$ coincides with 
the number of zeros of $\tau$ in $[0,2\pi)$.
The latter is equal to two.
By Theorem \ref{deficiency indices unitary} (3), 
we have $d_{\pm}(\hat{T}_j)=2$ for each $j=1,2$, and hence $d_{\pm}(T)=4$ holds.
Corollary \ref{corspec} (2) implies that $\sigma(T)=\mathbb{C}$. 
It follows from Remark \ref{conclusion from lemma} that no self-adjoint extensions of $T$ are strong time operators of $U$.

(2) This follows from Theorem \ref{TimeopofQW} (2) and the fact that $\sigma(P) = \mathbb{Z}$.
\end{proof}

\begin{example}\label{hadamard}
We consider the Hadamard walk, which is the case $\displaystyle C=\frac{1}{\sqrt{2}}\begin{pmatrix}
1&1\\
1&-1
\end{pmatrix}$.
Then we see that
\[
\lambda_1(k) = \frac{\sqrt{1+\cos^2{k}}+i\sin{k}}{\sqrt{2}},\ \ \ \ \ \lambda_2(k) = \frac{-\sqrt{1+\cos^2{k}}+i\sin{k}}{\sqrt{2}},
\]
whence
\[
\hat{T}_1 = \frac{1}{2}\left(\frac{\sqrt{1+\cos^2{k}}}{\cos{k}}P+P\frac{\sqrt{1+\cos^2{k}}}{\cos{k}}\right),\ \ \ \ \ \hat{T}_2=-\hat{T}_1.
\]
By Theorem \ref{deficiency indices unitary} (3), we get $d_{\pm}(\hat{T}_1)=d_{\pm}(\hat{T}_2)=2$, and hence $d_{\pm}(T)=4$.
Note that, in the present case, we have $m_1=m_2=0$, $\theta_1(k)=\arcsin{(\sin{k}/\sqrt{2})}$ and $\theta_2(k)=-\theta_1(k)+\pi$.
Thus, there exists a unitary operator 
\begin{align*}
V:\mathcal{H}=\ell^2(\mathbb{Z};\mathbb{C}^2)\to L^2&\left(\left[-\frac{\pi}{4},\frac{\pi}{4}\right],dE\right)\oplus L^2\left(\left[-\frac{\pi}{4},\frac{\pi}{4}\right],dE\right)\\
&\oplus L^2\left(\left[\frac{3\pi}{4},\frac{5\pi}{4}\right],dE\right) \oplus L^2\left(\left[\frac{3\pi}{4},\frac{5\pi}{4}\right],dE\right)
\end{align*}
such that
\[
V\bar{T}V^{-1} = P_{[-\pi/4,\pi/4]} \oplus P_{[-\pi/4,\pi/4]} \oplus P_{[3\pi/4,5\pi/4]}\oplus P_{[3\pi/4,5\pi/4]}
\]
and
\[
VUV^{-1} = e^{iE}\oplus e^{iE}\oplus e^{iE}\oplus e^{iE}.
\]
We remark that all self-adjoint extensions of $T$ are time operators of $U$, but none of them are strong time operators of $U$.
The latter follows from the fact that
\[
\sigma(U)=\left\{e^{iE}\mid E\in[-\pi/4,\pi/4]\cup[3\pi/4,5\pi/4]\right\}
\]
and Proposition \ref{prS}.
\end{example}

\subsection{Three-step quantum walks}

Asb\'oth and Obuse \cite{Multi1} introduced multi-step quantum walks (see also \cite{Multi2}) 
and calculated their winding numbers that differ from the winding numbers defined in this paper.
Here we consider a there-step quantum walk,
whose time evolution operator $U_3$ 
on $\mathcal{H}=\ell^2(\mathbb{Z};\mathbb{C}^2)$ is defined by
\[
U_3:= SC_1 S C_2 S C_3,
\]
where $S$ is the shift operator defined in Subsection \ref{construction discrete} and the coin operators are
\[ C_1 :=\begin{pmatrix}  1 & 0 \\ 0 & 1 \end{pmatrix}, 
\quad C_2 := \begin{pmatrix}  b & a \\ -a & b\end{pmatrix},
\quad C_3 := \begin{pmatrix}  b & -a \\ a & b\end{pmatrix}
\]
with $a,b \in \mathbb{R}$ and $a^2+b^2=1$.
In this subsection, we first construct a strong time operator of $U_3$ by a similar argument to Subsection \ref{construction discrete}.
We next compute its deficiency indices and spectrum as an application of Theorem \ref{deficiency indices unitary} and Corollary \ref{corspec}.
It turns out that $U_3$ has a self-adjoint strong time operator under a suitable condition.
This is because that a three-step quantum  walk has a non-zero winding number, thus Theorem \ref{deficiency indices unitary} (1) could be applied.

Let $\mathscr{F}$ be the Fourier transform defined in Subsection \ref{construction discrete}.
Then $\mathscr{F}U_3\mathscr{F}^{-1}$ is the multiplication operator on $\mathcal{K}$ by
\[
\hat{U}_3(k):=
\begin{pmatrix}
a^2e^{ik}+b^2e^{i3k} & abe^{ik}-abe^{i3k}\\
-abe^{-ik}+abe^{-i3k} & a^2e^{-ik}+b^2e^{-i3k}
\end{pmatrix},\ \ \ \ \ k\in [0,2\pi].
\]
The case $b^2=1$ is trivial.
Indeed, since $\hat{U}_3(k)=e^{i3k}\oplus e^{-i3k}$, 
Theorem \ref{lemma unitary} implies that $\hat{T}:=P/3\oplus(-P/3)$ is a strong time operator of $\hat{U}_3(k)$, 
and hence $T:=\mathscr{F}^{-1}\hat{T}\mathscr{F}$ is a strong time operator of $U_3$.
By Theorem \ref{deficiency indices unitary} (1), there exists a unitary operator $V:\mathcal{H}\to\oplus^6\mathcal{K}_1$ such that
\[
VTV^{-1}=\bigoplus_{j=0}^{2}\bigoplus^2\left(P+\frac{j}{3}\right), \ \ \ \ \ VU_3V^{-1}=\bigoplus^6e^{ik},
\]
where we have used the fact that the winding number of $e^{\pm i3k}$ is three.
Note that the above $T$ is self-adjoint and $\sigma(T)=\{s/3 \mid s\in\mathbb{Z}\}$.

In what follows, we concentrate on the case $0\leq b^2<1$.
For each $k\in\mathbb{R}$, the unitary matrix $\hat{U}_3(k)$ has exactly two eigenvalues $\lambda_1(k), \lambda_2(k)$.
They are given by
\[
\lambda_1(k) = \cos{k}(1-4b^2\sin^2{k})+i\sin{k}\sqrt{1+(8b^2-16b^4)\cos^2{k}+16b^4\cos^4{k}}
\]
and $\lambda_2(k) = \overline{\lambda_1(k)}$.
Since the inside of the square root is strictly positive, the map $\mathbb{R}\ni k\mapsto \lambda_j(k)\in \mathbb{C}$
is an analytic periodic function with period $2\pi$ for each $j=1,2$.
In particular, each $\lambda'_j$ has at most finitely many zeros in $[0,2\pi)$.
Let $v_1(k), v_2(k)$ be the corresponding normalized mutually orthogonal eigenvectors.
By a direct computation, we may assume that the map
$\mathbb{R}\ni k\mapsto v_j(k)\in \mathbb{C}^2$
is a Borel periodic function with period $2\pi$ for each $j=1,2$.
Define a $2\times 2$ unitary matrix $W(k)$ by $W(k):=\bigl(v_1(k),v_2(k)\bigr)$.
Then we obtain
\[
W(k)^{-1}\hat{U}_3(k)W(k) = 
\begin{pmatrix}
\lambda_1(k)&0\\
0&\lambda_2(k)
\end{pmatrix},\ \ \ \ \ k\in [0,2\pi].
\]
By Theorem \ref{lemma unitary}, we get a strong time operator
\[
\hat{T}_j:=\frac{i}{2}\left(\frac{\lambda_j(k)}{\lambda'_j(k)}P+P\frac{\lambda_j(k)}{\lambda'_j(k)}\right)
\]
of the multiplication operator $\lambda_j(k)$ for each $j=1,2$.
We use $W$ to denote the multiplication operator on $\mathcal{K}$ by $W(k)$ and set
\[
T :=\mathscr{F}^{-1}W
\begin{pmatrix}
\hat{T}_1&0\\
0&\hat{T}_2
\end{pmatrix}
W^{-1}\mathscr{F}.
\]
Then $T$ is a strong time operator of $U_3$.

Following Theorem \ref{deficiency indices unitary}, we next compute the number of zeros of $\lambda'_j$ in $[0,2\pi)$ for each $j=1,2$.
Note that we concentrate on the case $0\leq b^2<1$.
Put
\[
\alpha(k):=\cos{k}(1-4b^2\sin^2{k}),\ \ \ \ \ \beta(k):=1+(8b^2-16b^4)\cos^2{k}+16b^4\cos^4{k}
\]
for any $k\in\mathbb{R}$.
These functions satisfy $\beta(k)\sin^2{k}=1-\alpha(k)^2$ for every $k\in\mathbb{R}$.
In particular, $1-\alpha(k)^2=0$ if and only if $k\in\pi\mathbb{Z}$ because $\beta$ is strictly positive.
Let
\[
\epsilon(k):={\rm sgn}(\sin{k}),\ \ \ \ \
k\not\in\pi \mathbb{Z}.
\]
If $k\not\in\pi \mathbb{Z}$, we can write $\lambda_1(k)$ as
\[
\lambda_1(k)=\alpha(k)+i\sin{k}\sqrt{\beta(k)} = \alpha(k)+i\epsilon(k)\sqrt{1-\alpha(k)^2},
\]
and thus
\[
\frac{\lambda'_1(k)}{\lambda_1(k)} 
= \frac{\alpha'(k)\left\{1-i\epsilon(k)\frac{\alpha(k)}{\sqrt{1-\alpha(k)^2}}\right\}}{\alpha(k)+i\epsilon(k)\sqrt{1-\alpha(k)^2}} 
\times \frac{i\epsilon(k)\sqrt{1-\alpha(k)^2}}{i\epsilon(k)\sqrt{1-\alpha(k)^2}}
= \frac{\alpha'(k)}{i\epsilon(k)\sqrt{1-\alpha(k)^2}}.
\]
By taking absolute values of both sides, we obtain 
\[
|\lambda'_1(k)| = \frac{|\alpha'(k)|}{\sqrt{1-\alpha(k)^2}} = \frac{|1+8b^2-12b^2\sin^2{k}|}{\sqrt{\beta(k)}},\ \ \ \ \ k\in\mathbb{R}.
\]
We conclude that
\[
\textrm{the number of zeros of $\lambda'_1$ in $[0,2\pi)$} = \begin{cases} 0,\ \ \ \ \ &0\leq b^2<1/4,\\
2, &b^2=1/4,\\
4, &1/4<b^2<1.
\end{cases}
\]
Since $\lambda_2(k) = \overline{\lambda_1(k)}$ for any $k\in\mathbb{R}$, 
the number of zeros of $\lambda'_2$ in $[0,2\pi)$ is equal to the number of zeros of $\lambda'_1$ in $[0,2\pi)$.

Before applying Theorem \ref{deficiency indices unitary}, let us compute the winding number of $\lambda_j$ for each $j=1,2$.
We still concentrate on the case $0\leq b^2<1$.
By the above computation, we get
\begin{align*}
\textrm{the winding number of $\lambda_1$} &= \frac{1}{2\pi i}\int_0^{2\pi}\frac{\lambda'_1(k)}{\lambda_1(k)}\,dk\\
&= \frac{1}{2\pi i}\int_0^{\pi}\frac{\lambda'_1(k)}{\lambda_1(k)}\,dk+\frac{1}{2\pi i}\int_{\pi}^{2\pi}\frac{\lambda'_1(k)}{\lambda_1(k)}\,dk\\
&= -\frac{1}{2\pi}\int_0^{\pi}\frac{\alpha'(k)}{\sqrt{1-\alpha(k)^2}}\,dk+\frac{1}{2\pi}\int_{\pi}^{2\pi}\frac{\alpha'(k)}{\sqrt{1-\alpha(k)^2}}\,dk\\
&= \frac{1}{2\pi}\int_0^{\pi}\frac{d}{dk}\arccos \alpha(k)\,dk-\frac{1}{2\pi}\int_\pi^{2\pi}\frac{d}{dk}\arccos \alpha(k)\,dk\\
&= 1.
\end{align*}
Since $\lambda_2(k) = \overline{\lambda_1(k)}$, it follows that the winding number of $\lambda_2$ is minus one.

Summing up the above arguments, we obtain the following theorem.

\begin{theorem}\label{final}
Let $0\leq b^2<1$.
For the strong time operator $T$ of $U_3$ defined above, the following statements hold.
\begin{itemize}{}{}
\item[(1)] If $0\leq b^2<1/4$, then $T$ is self-adjoint and $\sigma(T)=\mathbb{Z}$. 
Moreover, there exists a unitary operator $V:\mathcal{H}\to\oplus^2\mathcal{K}_1$ such that
\[
VTV^{-1}=\bigoplus^2P,\ \ \ \ \ VU_3V^{-1}=\bigoplus^2e^{ik}.
\]
In particular, $T$ is a self-adjoint strong time operator of $U_3$.
\item[(2)] If $b^2=1/4$, then $d_{\pm}(T)=4$ and $\sigma(T)=\mathbb{C}$.
\item[(3)] If $1/4<b^2<1$, then $d_{\pm}(T)=8$ and $\sigma(T)=\mathbb{C}$.
\end{itemize}
In the case of (2) or (3), $T$ has self-adjoint extensions.
All of them are time operators of $U_3$.
\end{theorem}

\begin{remark}
In the case of (2) or (3) in Theorem \ref{final}, it is natural to ask whether 
self-adjoint extensions of $T$ are strong time operators of $U_3$.
This problem is left for future study.
\end{remark}

\section*{Acknowledgement}
This work was supported by JSPS KAKENHI (Grant Number JP18K03327, JP16K17612 and 26800055), 
and by the Research Institute for Mathematical Sciences, a Joint Usage/ Research Center located in Kyoto University.

\end{document}